\documentclass[12pt]{article}
\usepackage{natbib}
\bibpunct[,]{(}{)}{;}{a}{,}{,}
\usepackage[utf8]{inputenc} 
\usepackage[T1]{fontenc}    
\usepackage{hyperref}       
\usepackage{url}            
\usepackage{booktabs}       
\usepackage{amsfonts}       
\usepackage{microtype}      

\usepackage{amsthm}
\usepackage{amsmath,amssymb}
\usepackage[pdftex]{graphicx}

\usepackage{algorithm, algpseudocode}
\newtheorem*{theorem*}{Theorem}
\newtheorem*{definition*}{Definition}
\newtheorem*{lemma*}{Lemma}
\newtheorem*{example*}{Example}
\newtheorem*{proposition*}{Propsotion}
\newtheorem*{corollary*}{Corollary*}
\newtheorem*{remark}{Remark}

\newtheorem{theorem}{Theorem}

\graphicspath{{figs/}}
\makeatletter

\newcommand{\kakko}[1]{\left( #1 \right)}
\newcommand{\ckakko}[1]{\left\{ #1 \right\}}
\newcommand{\dkakko}[1]{\left[ #1 \right]}

\newcommand{\nkakko}[1]{\left\| #1 \right\|}

\newcommand{\D}{\mathrm{d}}

\DeclareMathOperator*{\minimize}{\text{minimize} \quad}

\newcommand{\rank}[1]{\operatorname{rank}\left( #1 \right)}
\newcommand{\sign}[1]{\operatorname{sign} ( #1 )}

\newcommand{\abs}[1]{| #1 |}
\newcommand{\group}{{\mathcal{G}}}

\newcommand{\betah}{\hat{\beta}}
\newcommand{\betag}{\beta^{\mathcal{G}}}
\newcommand{\betahg}{\hat{\beta}^{\mathcal{G}}}
\newcommand{\betaz}{\beta^{\mathcal{G}}_{-0}}

\newcommand{\lamg}{\lambda^{\mathcal{G}}}
\newcommand{\lamb}{\bar{\lambda}^{\mathcal{G}}}
\newcommand{\Xg}{X^{\mathcal{G}}}
\newcommand{\xg}{x^{\mathcal{G}}}
\newcommand{\Xz}{X^{\mathcal{G}}_{-0}}

\newcommand{\Xtz}{X^{\mathcal{G}}_{-0}}

\newcommand{\gmax}{\overline{g}}

\newcommand{\blindCode}{0}

\begin{document}

\def\spacingset#1{\renewcommand{\baselinestretch}%
{#1}\small\normalsize} \spacingset{1}


\if0\blindCode
{
  \title{\bf An Exact Solution Path Algorithm for SLOPE and Quasi-Spherical OSCAR}
  \author{Shunichi Nomura\footnote{Email:nomura@ism.ac.jp} \\
    The Institute of Statistical Mathematics, 10-3 Midoricho, Tachikawa, Tokyo, Japan}
  \maketitle
} \fi

\if1\blindCode
{
  \bigskip
  \bigskip
  \bigskip
  \begin{center}
    {\LARGE\bf An Exact Solution Path Algorithm for SLOPE and Quasi-Spherical OSCAR}
\end{center}
  \medskip
} \fi

\bigskip
\begin{abstract}
Sorted $L_1$ penalization estimator (SLOPE) is a regularization technique for sorted absolute coefficients in high-dimensional regression.
By arbitrarily setting its regularization weights $\lambda$ under the monotonicity constraint,
SLOPE can have various feature selection and clustering properties.
On weight tuning, the selected features and their clusters are very sensitive to the tuning parameters.
Moreover, the exhaustive tracking of their changes is difficult using grid search methods. 
This study presents a solution path algorithm that provides the complete and exact path of solutions for SLOPE 
in fine-tuning regularization weights.
A simple optimality condition for SLOPE is derived and used to specify the next splitting point of the solution path.
This study also proposes a new design of a regularization sequence $\lambda$ for feature clustering,
which is called the quasi-spherical and octagonal shrinkage and clustering algorithm for regression (QS-OSCAR).
QS-OSCAR is designed with a contour surface of the regularization terms most similar to a sphere.
Among several regularization sequence designs, sparsity and clustering performance are compared through simulation studies.
The numerical observations show that QS-OSCAR performs feature clustering more efficiently than other designs.
\end{abstract}

\newpage
\spacingset{1.5} 
\section{Introduction}
This study considers feature selection and clustering in high-dimensional regression analysis.
The least absolute shrinkage and selection operator (Lasso)~\citep{tibshirani1996regression} is a $L_1$ regularization technique
that shrinks the least-square estimator (LSE) and removes predictors with slight effects on response by leading their coefficients to exact zero.
The Lasso is often combined with $L_2$ regularization terms that shrink the LSE as well and 
make the coefficients of highly correlated predictors similar in absolute value.
The $L_1$ penalization technique in the Lasso is extended to the fused Lasso~\citep{tibshirani2005sparsity}
that can be used for clustering the predictors into groups sharing common coefficient values.

The sorted $L_1$ penalization estimator (SLOPE)~\citep{bogdan2013, bogdan2015} or ordered weighted L1 (OWL)~\citep{zeng2014}, 
which is also a generalization of the Lasso and hereafter called SLOPE, is dealt with herein.
Let us consider the high-dimensional linear regression $y = X\beta + w$, where $y \in \mathbb{R}^n$ is a response vector;
$X \in \mathbb{R}^{n\times p}$ is a design matrix; $\beta \in \mathbb{R}^p$ is a coefficient vector; and $w$ is an observation noise vector.
SLOPE is a regularization technique on the order statistics $|\beta|_{(1)} \le |\beta|_{(2)} \le \cdots \le |\beta|_{(p)}$ of the absolute values of the coordinates in $\beta$
and is formulated as follows:
\begin{equation} \label{slope}
\minimize_{\beta\in \mathbb{R}^p}\ \ \frac{1}{2} \nkakko{y-X\beta}^2 + \sum_{i=1}^p \lambda_i |\beta|_{(i)}, 
\end{equation}
where $\lambda = (\lambda_1,\lambda_2,\dots,\lambda_p)^\top$ denotes the regularization parameters in a non-decreasing order, such that
$0\le \lambda_1 \le \lambda_2 \le \cdots \le \lambda_p$.
The monotonic structure of $L_1$ penalties on the sorted coefficients simultaneously encourages the sparsity and grouping of absolute coefficients.

The strength and balance of feature selection and clustering depend on the design of the regularization sequence,
which can be arbitrarily set under its monotonicity constraint $0\le \lambda_1 \le \lambda_2 \le \cdots \le \lambda_p$.
\citep{bogdan2015} presents two types of regularization sequences
designed to control the false discovery rate in feature selection under two different assumptions on $X$.
\citep{bondell2008simultaneous} proposed the octagonal shrinkage and clustering algorithm for regression (OSCAR)
for feature clustering with regularization terms originally defined with $L_\infty$ penalties, but can be represented by SLOPE as well.
Some recent studies have proven that SLOPE can achieve the minimax rate of convergence under certain sparsity setups in high-dimensional regression
~\citep{su2016, lecue2018, bellec2018}.

SLOPE \eqref{slope} is a convex optimization problem with efficient solving algorithms \citet{bogdan2015, luo2019efficient}.
\citet{bao2020fast} and \citet{larsson2020} proposed the screening rules for SLOPE to fasten updating solutions with a consecutive sequence of regularization weights.
However, since feature selection and clustering are very sensitive to the regularization weights,
tracking all the changes in the feature grouping of a continuous path using grid search methods is difficult.
\citet{takahashi2020} showed that OSCAR has more than 1000 changing points in its solution path for 100 coefficients in its numerical experiment.
\citet{bao2019} recently proposed an approximate path algorithm for SLOPE to obtain an approximate solution path with an arbitrary accuracy bound,
which does not cover all the changing points either.
Although some exact path algorithms have been proposed for the Lasso and more general $L_1$ regularized regression \citep{rosset2007, tibshirani2011},
they cannot be applied to SLOPE owing to its non-separable structure.

This study proposes an exact path algorithm for SLOPE, which is extended from that for OSCAR~\citep{takahashi2020}.
The proposed algorithm detects all the changing points in a solution path where the coefficient groups fuse and split.
A general optimality condition for SLOPE is derived to specify the splitting point of a path.
This condition involves the order of elements in the gradient of the loss function,
whose changes are also monitored as the switching events of a path.

After presenting the path algorithm, a new design for the regularization weights $\lambda$, called 
the quasi-spherical and octagonal shrinkage and clustering algorithm for regression (QS-OSCAR), is proposed as an OSCAR modification
for a more efficient feature clustering.
The regularization weights in QS-OSCAR are designed to have its octagonal contour surface most similar to a sphere on the space $\mathbb{R}^p$
in the sense of the ratio of the circumradius and inradius of the surface.
The features of QS-OSCAR and other types of regularization sequences for SLOPE are compared in numerical experiments.

\section{Path algorithm for SLOPE}
This section presents a new solution path algorithm for SLOPE.
Since this algorithm is based on the path algorithm for OSCAR~\citep{takahashi2020},
this section proceeds in the analogue of Section 3 of \citet{takahashi2020}.

\subsection{Notation}
Let us consider a solution path $\beta(\eta)$ for SLOPE \eqref{slope} with respect to the regularization parameters $\lambda(\eta) = \lambda_0 + \eta \bar{\lambda}$
driven by a scalar hyperparameter $\eta \in [0,\eta_{\max})$ with its initial point $\lambda_0 = (\lambda_{01},\lambda_{02},\dots,\lambda_{0p})^\top$
and search direction $\bar{\lambda} = (\bar{\lambda}_1,\bar{\lambda}_2,\dots,\bar{\lambda}_p)^\top \neq 0$. 
The terminal value $\eta_{\max}$ of the parameter $\eta$ may be either finite or infinite
as long as it holds that $0\le \lambda_1(\eta) \le \cdots \le \lambda_p(\eta)$ for any $\eta \in [0,\eta_{\max})$. 
Accordingly, we assume $\rank{X} = p$, which implies that $n \ge p$ and the objective function \eqref{slope} is strictly convex
(otherwise we can add a $L_2$ penalty term with a small weight to \eqref{slope} to make the objective function strictly convex).

The regularization terms in \eqref{slope} encourage the absolute values of the coefficients $|\betah(\eta)|$ to be zero and equal to each other.
Therefore, for a fixed hyperparameter $\eta \in [0,\eta_{\max})$, 
we define the set of fused groups $\group(\eta) = \{G_0, G_1, \dots, G_{\gmax}\}$ and
the grouped absolute coefficients $\betag(\eta) = (\betag_0, \betag_1,\dots,\betag_{\gmax})^\top \in \mathbb{R}^{\gmax+1}$ 
to satisfy the following statements:
\begin{itemize}
                \item $\bigcup_{g=0}^{\gmax} G_g = \ckakko{1, \cdots, p}$, where $G_0$ may be an empty set, but others may not.
                \item $0 = \betag_0< \betag_1 < \cdots < \betag_{\gmax}$ and $| \beta_i | = \betag_g $ for $i\in G_g$.
\end{itemize} 
Let $p_g$ be the cardinality of $G_g$ and $q_g = p_0 + \cdots + p_{g-1}$.
Corresponding to the fused groups, 
we define the grouped design matrix by
$\Xg = (\xg_0, \xg_1,\dots,\xg_{\gmax}) \in \mathbb{R}^{n\times (\gmax+1)}$, 
where $\xg_g = \sum_{j \in G_g} \sign{\beta_j} x_j$, and $x_j$ is the $j$-th column vector of $X$.
Let us denote by $\betaz(\eta)= (\betag_1,\dots,\betag_{\gmax})^\top$
the nonzero elements of the grouped coefficients and by $\Xz = (\xg_1,\dots,\xg_{\gmax}) $ 
the corresponding columns of the grouped design matrix.

\subsection{Piecewise linear solution path and optimal condition} 
Given the set of fused groups $\group(\eta)$, where $\gmax \ge 1$,
the problem \eqref{slope} can be reduced to the following quadratic programming problem:
\begin{equation*} 
\minimize_{\betaz\in \mathbb{R}^{\gmax}}\ \ \frac{1}{2} \nkakko{y - \Xz \betaz}^2 
+ \sum_{g=1}^{\gmax} \lamg_{0g} \betag_g 
+ \eta \sum_{g=1}^{\gmax} \lamb_g \betag_g,
\end{equation*}
where $\lamg_{0g} = \sum_{i = q_g + 1}^{q_{g+1}} \lambda_{0i}$ and $\lamb_g = \sum_{i = q_g + 1}^{q_{g+1}} \bar{\lambda}_i$.
Therefore, until the set of fused groups $\group(\eta)$ changes,
the solution $\beta(\eta)$ takes a linear path with $\betag(\eta)_0 \equiv 0$ and nonzero grouped elements:
\begin{equation} \label{solution}
\betaz(\eta) = \dkakko{\kakko{\Xtz}^\top \Xtz}^{-1} \dkakko{\kakko{\Xtz}^\top  y + \lamg_0 + \eta \lamb}, 
\end{equation}
where $\lamg_0 = (\lamg_{01},\lamg_{02},\dots,\lamg_{0\gmax})^\top$ and $\lamb = (\lamb_{1},\lamb_{2},\dots,\lamb_{\gmax})^\top$.

We consider the optimality condition for \eqref{slope} to check whether $\beta(\eta)$ defined by \eqref{solution} is optimal when $\eta$ moves.
Here, we treat a slightly more general setting in which
the square loss function $\frac{1}{2} \| y - X \beta \|^2$ in \eqref{slope} is generalized to a strongly convex and differentiable loss function $f(\beta;y)$.
Then, the objective function is described by
\begin{equation} \label{slopef}
\minimize_{\beta\in \mathbb{R}^p}\ \ f(\beta;y) + \sum_{i=1}^p \lambda_i |\beta|_{(i)}.
\end{equation}
and has a unique global optimum $\betah$ because of the strong convexity of $f$.
Let $\nabla f(\beta) = (\nabla f_1(\beta),\nabla f_2(\beta),\dots, \nabla f_p(\beta))^\top$ be a gradient of the loss function $f$
and $o(1),o(2),\dots,o(p) \in \{1,\dots,p\}$ be the order of indices in $\betah$ 
such that $G_g = \{ o(q_g+1),o(q_g+2),\dots,o(q_{g+1}) \} $ and 
$s_{o(q_g+1)} \nabla f_{o(q_g+1)}(\betah) \le s_{o(q_g+2)} \nabla f_{o(q_g+2)}(\betah) \le \cdots \le s_{o(q_g+1)} \nabla f_{o(q_g+1)}(\betah)$ 
for each group $g$, where $s = (s_1,s_2,\dots,s_p)^\top$ is defined by
\begin{equation*} 
s_i = \begin{cases}  -\sign{\betah_i} & \text{if $\betah_i \neq 0$}, \\ \sign{\nabla f_i(\betah)} & \text{if $\betah_i = 0$}. \end{cases}
\end{equation*} 
Define $\nabla f^{\group}_{g,k}(\betah) = \sum_{i=q_g + k}^{q_{g+1}} s_{o(i)} \nabla f_{o(i)}(\betah)$ 
and $\lamg_{g,k} = \sum_{i=q_g + k}^{q_{g+1}} \lambda_i$ for $g=1,\dots,\gmax$, $k=1,\dots,p_g$.
The optimality condition of $\betah$ for \label{slopef} can be given by the following theorem:
\begin{theorem} \label{theorem}
The coefficient vector $\betah$ is a global optimum of \eqref{slopef} if and only if the following conditions hold:
\begin{eqnarray}\label{cond1}
\lamg_{g,1} - \nabla f^{\group}_{g,1}(\betah) & = & 0 \qquad g = 1,\dots,\gmax, \\ 
\lamg_{0,k} - \nabla f^{\group}_{0,k}(\betah) & \ge & 0 \qquad k = 1,\dots,p_0, \label{cond2}  \\ 
\lamg_{g,k} - \nabla f^{\group}_{g,k}(\betah) & \ge & 0 \qquad g = 1,\dots,\gmax,\; k = 2,\dots,p_g. \label{cond3}
\end{eqnarray}
\end{theorem}

\begin{proof}
First, we prove the sufficiency of the conditions \eqref{cond1}, \eqref{cond2} and \eqref{cond3} for the optimality of $\betah$.
The objective function $h(\beta) = f(\beta;y) + \sum_{i=1}^p \lambda_i |\beta|_{(i)}$ is strictly convex, and hence has a unique local optimum.
Therefore, it suffices to prove that $h(\tilde{\beta}) \ge h(\betah)$ if $\| \tilde{\beta}-\betah \|_{\infty} < \frac{1}{2} \min_{g=1,\dots,\gmax} (\betahg_g - \betahg_{g-1})$.
Define $\Delta = |\tilde{\beta}| - |\betah|$ and the order $\tilde{o}(1),\tilde{o}(2),\dots,\tilde{o}(p) \in \{1,\dots,p\}$ of its coordinates 
such that $G_g = \{ \tilde{o}(q_g+1),\tilde{o}(q_g+2),\dots,\tilde{o}(q_{g+1}) \} $ and 
$\Delta_{\tilde{o}(q_g+1)} \le \Delta_{\tilde{o}(q_g+2)} \le \cdots \le \Delta_{\tilde{o}(q_{g+1})} $. 
Because the order satisfies $0 \le |\tilde{\beta}_{\tilde{o}(1)}| \le |\tilde{\beta}_{\tilde{o}(2)}| \le \cdots \le |\tilde{\beta}_{\tilde{o}(p)}|$ and
$0 \le |\betah_{\tilde{o}(1)}| \le |\betah_{\tilde{o}(2)}| \le \cdots \le |\betah_{\tilde{o}(p)}|$,
we have
\begin{equation*} 
\sum_{i=1}^p \lambda_i (|\tilde{\beta}|_{(i)} - |\betah|_{(i)}) 
= \sum_{i=1}^p \lambda_i (|\tilde{\beta}_{\tilde{o}(i)}| - |\betah_{\tilde{o}(i)}|) = \sum_{i=1}^p \lambda_i \Delta_{\tilde{o}(i)}.
\end{equation*} 
Consequently, we obtain
\begin{align*} 
h(\tilde{\beta}) - h(\betah) 
&\ge \sum_{i=1}^p \nabla f_{\tilde{o}(i)}(\betah) (\tilde{\beta}_{\tilde{o}(i)}-\betah_{\tilde{o}(i)}) + \sum_{i=1}^p \lambda_i \Delta_{\tilde{o}(i)} \\
&\ge - \sum_{g=0}^{\gmax} \sum_{i=q_g + 1}^{q_{g+1}} s_{\tilde{o}(i)} \nabla f_{\tilde{o}(i)}(\betah) \Delta_{\tilde{o}(i)} + \sum_{i=1}^p \lambda_i \Delta_{\tilde{o}(i)},
\end{align*} 
where the first inequality follows from the convexity of $f$, while the second one follows from the equalities that
$- s_i \nabla f_{\tilde{o}(i)}(\betah) = -|\nabla f_{\tilde{o}(i)}(\betah)|$ for $i = 1,\dots,p_0$
and $\Delta_i = - s_i(\tilde{\beta}_i-\betah_i)$ for $i = p_0+1,\dots, p$.
Since we have $s_{o(q_g+1)} \nabla f_{o(q_g+1)}(\betah) \le s_{o(q_g+2)} \nabla f_{o(q_g+2)}(\betah) \le \cdots \le s_{o(q_{g+1})} \nabla f_{o(q_{g+1})}(\betah)$ 
for $g=0,\dots,\gmax$ by the definition of $o(1),o(2),\dots,o(p)$,
the Hardy-Littlewood-P\'{o}lya inequality (Theorem 368 in ~\citet{hardy1934}) yields that
$\sum_{i=q_g + 1}^{q_{g+1}} s_{\tilde{o}(i)} \nabla f_{\tilde{o}(i)}(\betah) \Delta_{\tilde{o}(i)} 
\le \sum_{i=q_g + 1}^{q_{g+1}} s_{o(i)} \nabla f_{o(i)}(\betah) \Delta_{\tilde{o}(i)}$.
Above all, we obtain
\begin{align*} 
h(\tilde{\beta}) - h(\betah) 
&\ge - \sum_{g=0}^{\gmax} \sum_{i=q_g + 1}^{q_{g+1}} s_{o(i)} \nabla f_{o(i)}(\betah) \Delta_{\tilde{o}(i)} + \sum_{i=1}^p \lambda_i \Delta_{\tilde{o}(i)} \\
&= \sum_{g=0}^{\gmax} \sum_{i=q_g + 1}^{q_{g+1}} (\lambda_i - s_{o(i)} \nabla f_{o(i)}(\betah)) \Delta_{\tilde{o}(i)} \\
&= \sum_{g=0}^{\gmax} \{ (\lamg_{g,1} - \nabla f^{\group}_{g,1}(\betah))\Delta_{\tilde{o}(q_g + 1)} 
+ \sum_{k=2}^{p_g} \{ (\lamg_{g,k} - \nabla f^{\group}_{g,k}(\betah))(\Delta_{\tilde{o}(q_g + k)}-\Delta_{\tilde{o}(q_g + k - 1)})  \}.
\end{align*} 
Since $\Delta_{\tilde{o}(1)} \ge 0$ if $p_0 \ge 1$ and $\Delta_{\tilde{o}(q_g + k)}-\Delta_{\tilde{o}(q_g + k - 1)} > 0$,
$h(\tilde{\beta}) - h(\betah) \ge 0$ for any $\tilde{\beta}$ within $\| \tilde{\beta}-\betah \|_{\infty} < \frac{1}{2} \min_{g=1,\dots,\gmax} (\betahg_g - \betahg_{g-1})$
if the conditions \eqref{cond1}, \eqref{cond2} and \eqref{cond3} are all satisfied, which completes the proof of the sufficiency.

Next, we prove the necessity of the conditions \eqref{cond1}, \eqref{cond2} and \eqref{cond3} for the optimality of $\betah$.
If the condition \eqref{cond1} fails (i.e., $\lamg_{g,1} - \nabla f^{\group}_{g,1}(\betah) \neq 0$) for some $g=1,\dots, \gmax$, 
we define $\tilde{\beta}$ by
\begin{equation*} 
\tilde{\beta}_i = 
\left\{ \begin{array}{ll} \betah_i - \sign{ \lamg_{g,1} - \nabla f^{\group}_{g,1}(\betah) } \sign{\betah_i} \delta & i \in G_g, \\ \betah_i & \mbox{otherwise}, 
\end{array} \right.
\end{equation*} 
with a small positive scalar $0 < \delta < \min_{g=1,\dots,\gmax}(\betag_g - \betag_{g-1})$. We then have
\begin{align*} 
h(\tilde{\beta}) - h(\betah) &= f(\tilde{\beta}) - f(\betah) -  \sign{\lamg_{g,1} - \nabla f^{\group}_{g,1}(\betah)} \delta \sum_{i=q_g +1}^{q_{g+1}} \lambda_i  \\
& =  -|\lamg_{g,1} - \nabla f^{\group}_{g,1}(\betah)|\delta + o(\delta),
\end{align*} 
which takes a negative value for a sufficiently small $\delta$, indicating that $\betah$ is not optimal.
If either condition \eqref{cond2} or \eqref{cond3} fails (i.e. $\lamg_{g,k} - \nabla f^{\group}_{g,k}(\betah) < 0$) for some $g$ and $k$, 
we define $\tilde{\beta}$ by
\begin{equation*} 
\tilde{\beta}_{o(i)} =
\left\{  \begin{array}{ll} \betah_{o(i)} + \sign{\betah_{o(i)}} \delta & i \in \{q_g + k,\dots, q_{g+1}\}, \\ \betah_{o(i)} & \mbox{otherwise}, \end{array} \right.
\end{equation*} 
with a small positive scalar $0 < \delta < \min_{g=1,\dots,\gmax}(\betag_g - \betag_{g-1})$. We then have
\begin{align*} 
h(\tilde{\beta}) - h(\betah) &= f(\tilde{\beta}) - f(\betah) + \delta \sum_{i=q_g +1}^{q_{g+1}} \lambda_i \\
& =  \{ \lamg_{g,k} - \nabla f^{\group}_{g,k}(\betah) \} \delta + o(\delta),
\end{align*} 
which takes a negative value for a sufficiently small $\delta$, indicating that $\betah$ is not optimal.
Thus, $\betah$ is not optimal unless the conditions \eqref{cond1}, \eqref{cond2} and \eqref{cond3} are all satisfied, which completes the proof.
\end{proof}

\subsection{Path algorithm} 
In this subsection, we construct a path algorithm for SLOPE.
As in the path algorithm for OSCAR \citep{takahashi2020}, the proposed algorithm monitors three kinds of events, that is, fusing, splitting, and switching events in its solution path.
Before showing the algorithm, we introduce these events and specify their occurrence times in turn.

First, the fusing event is a type of event where adjacent groups are fused when their coefficients collide.
Since the grouped absolute coefficients $\betag(\eta)$ move linearly as in \eqref{solution} along $\eta$
with its slope $\frac{\D \betag}{\D \eta} = \kakko{0, (\lamb)^\top [(\Xz)^{\top} \Xz]^{-1}}^\top$, 
two adjacent groups $G_g$ and $G_{g+1}$ have to be fused at time $\Delta_{g}^{\text{fuse}}$ from $\eta$ given by
\begin{equation} \label{fuse}
\Delta^{\text{fuse}}_g(\eta) = \begin{cases}  \frac{-\betag_{g+1}(\eta) + \betag_{g}(\eta)}{\frac{\D \betag_{g+1}}{\D \eta} - \frac{\D \betag_{g}}{\D \eta}} & \text{if $\frac{\D \betag_{g+1}}{\D \eta} < \frac{\D \betag_{g}}{\D \eta}$}, \\
\displaystyle \infty & \text{otherwise}.
\end{cases}
\end{equation}

Next, the splitting event is a type of event where a group is split into two groups 
when either condition \eqref{cond2} or \eqref{cond3} is violated for some $g$ and $k$.
When it happens for some $g$ and $k$, the fused group $G_g = \{o(q_g+1),\dots,o(q_{g+1}) \}$ is split into
two groups $G_g = \{ o(q_g+1),\dots, o(q_g+k-1)\}$ and $G_{g+1} = \{ o(q_g+k),\dots,o(q_{g+1}) \}$.
The indices $g+1,\dots,\gmax$ of the subsequent groups are then increased by one.
Note that, if the condition \eqref{cond2} is violated for $k=1$, 
$G_0 = \{ o(1),\dots,o(p_0) \}$ is split into $G_0 = \emptyset$ and $G_1 = \{ o(1),\dots,o(p_0) \}$.

The optimality conditions \eqref{cond2} and \eqref{cond3} in Theorem~\ref{theorem} are used to specify the timings of splitting events. 
Here, we apply the square loss function $f(\beta;y) = \frac{1}{2} \| y-X\beta \|^2$ and its gradient $\nabla f(\beta) = -X^\top (y-X\beta) = X^\top (\Xz \betaz-y)$
to Theorem~\ref{theorem}.
Define the order $o(1),o(2),\dots,o(p) \in \{1,\dots,p\}$ of indices in $\beta$ to satisfy $G_g = \{ o(q_g+1),o(q_g+2),\dots,o(q_{g+1}) \} $ and 
$s_{o(q_g+1)} x_{o(q_g+1)}^\top (\Xz \betaz-y) \le s_{o(q_g+2)} x_{o(q_g+2)}^\top (\Xz \betaz-y) \le \cdots \le s_{o(q_{g+1})} x_{o(q_{g+1})}^\top (\Xz \betaz-y)$ 
for each group, where $s = (s_1,s_2,\dots,s_p)^\top$ is defined by
\begin{equation*} 
s_i = \begin{cases}  -\sign{\beta_i} & \text{if $\beta_i \neq 0$}, \\ \sign{x_i^\top (\Xz \betaz-y)} & \text{if $\beta_i = 0$}. \end{cases}
\end{equation*} 
Furthermore, we define $\nabla f^{\group}_{g,k}(\beta) = \sum_{i=q_g + k}^{q_{g+1}} s_{o(i)} \nabla f_{o(i)}(\beta) = x_{\overline{o}(g,k)}^\top (\Xz \betaz -y) $,
where $x_{\overline{o}(g,k)} = \sum_{j=q_g+k}^{q_{g+1}} s_{o(j)} x_{o(j)}$.
Thus, the splitting event occurs when either condition \eqref{cond2} or \eqref{cond3} is violated for some $g$ and $k$ at the time $\Delta_{g, k}^{\text{split}}$ from $\eta$ given by
\begin{equation} \label{split}
	\Delta_{g, k}^{\text{split}}(\eta) = \begin{cases} \frac{ x_{\overline{o}(g,k)}^\top \{y-\Xz \betaz(\eta)\} + \lamg_{g,k}(\eta) }{x_{\overline{o}(g,k)}^\top \frac{\D \betaz}{\D \eta}-\lamb_{g,k}} & \text{if $x_{\overline{o}(g,k)}^\top \frac{\D \betaz}{\D \eta}-\lamb_{g,k} > 0$}, \\
	\infty & \text{otherwise}, 
	\end{cases}
\end{equation}
where $\lamg_{g,k}(\eta) = \sum_{i=q_g + k}^{q_{g+1}} \lambda_i(\eta) = \sum_{i=q_g + k}^{q_{g+1}} (\lambda_{0i} + \eta \bar{\lambda}_i )$ and
$\lamb_{g,k} = \sum_{i=q_g + k}^{q_{g+1}}\bar{\lambda}_i$. 

Finally, the switching event is a type of event that changes the order $o(1),o(2),\dots,o(p) \in \{1,\dots,p\}$ of indices in $\beta$ or the sign of $s_{o(1)}$,
which may change the timings \eqref{split} of the splitting events.
Note that the slope $\frac{\D \betag}{\D \eta}$ of a solution path and the timings \eqref{fuse} of the fusing events are not affected by the switching events.
The switching event on the order of coefficients $o(1),\dots,o(p)$ occurs when the inequality 
$s_{o(k)} x_{o(k)}^\top (\Xz \betaz - y) \le s_{o(k+1)} x_{o(k+1)}^\top (\Xz \betaz - y)$ becomes reversed for some $k \in \{1,\dots,p-1\}\setminus \{q_1,\dots,q_{\gmax}\}$,
at which the indices assigned to $o(k)$ and $o(k+1)$ have to be switched.
The timing $\Delta_{k}^\text{switch}$ from $\eta$ of the event for each $k$ is given by
\begin{equation*} 
\Delta_{k}^\text{switch}(\eta) = \begin{cases} 
\frac{(s_{o(k)} x_{o(k)}^\top-s_{o(k+1)} x_{o(k+1)}^\top) \{y-\Xz \betaz(\eta)\}}{(s_{o(k)} x_{o(k)}^\top-s_{o(k+1)} x_{o(k+1)}^\top) \Xz \frac{\D \betaz}{\D \eta}} 
& \text{if $(s_{o(k)} x_{o(k)}^\top-s_{o(k+1)} x_{o(k+1)}^\top) \Xz \frac{\D \betaz}{\D \eta} > 0$}. \\
\infty & \text{otherwise}.
\end{cases}
\end{equation*}
Similarly, the switching event that reverses the sign of $s_{o(1)}$ occurs when $\beta_{o(1)} = 0$ and 
the sign $s_{o(1)} = \sign{\nabla f_i(\beta)} = \sign{x_{o(1)}^\top (\Xz \betaz - y)}$ reverses,
which occurs at time $\Delta_{0}^\text{switch}$ from $\eta$ given by
\begin{equation*} 
\Delta_{0}^\text{switch} = \begin{cases} \frac{x_{o(1)}^\top \{y-\Xz \betaz(\eta)\}}{x_{o(1)}^\top \Xz \frac{\D \betaz}{\D \eta}} 
& \text{if $s_{o(1)} x_{o(1)}^\top \Xz \frac{\D \betaz}{\D \eta} < 0$}. \\
\infty & \text{otherwise}.
\end{cases}
\end{equation*}

Algorithm~\ref{algorithm} outlines the proposed path algorithms for SLOPE using the events and their timings described above.
The initial solution for $\lambda = 0$ is obtained by the LSE $\beta^{(0)} = (X^\top X)^{-1}X^\top y$ if $\lambda_0 = 0$;
otherwise, we can use some efficient solvers~\citep{bogdan2015, luo2019efficient} to obtain the initial solution for SLOPE with $\lambda=\lambda_0$.
The algorithm iterates until the next event time $\eta^{(t)} + \Delta^{\min}$ becomes equal to or larger than $\eta_{\max}$.
The computational cost in each iteration can be evaluated in the same manner as in \citet{takahashi2020} and depends on which type of event occurs.
Each iteration where a fusing/splitting event occurs requires $\mathcal{O}(np)$ time by using the block matrix computation
In contrast, each iteration where a switching event occurs requires only $\mathcal{O}(n)$ time
because most of the timings of the next events can be updated by only subtracting $\Delta^{\min}$.
See Apendix C in \citet{takahashi2020} for more details of updating procedure in each event type.
Consequently, Algorithm~\ref{algorithm} requires $\mathcal{O}(T_\text{fuse+split}np + T_\text{switch} n)$ time except for the initial step to obtain $\beta^{(0)}$,
where $T_\text{fuse+split}$ and $T_\text{switch}$ are the numbers of iteration with fusing/splitting events and switching events, respectively.

\begin{algorithm} 
	\caption{Path algorithm for SLOPE} \label{algorithm}
	\begin{algorithmic}[1]
		\State $t \leftarrow 0$, $\eta^{(0)} \leftarrow 0$
                \State Solve \eqref{slope} with $\lambda = \lambda_0$ to obtain the initial solution $\beta^{(0)}$.
		\State Compute $\group$, $\betag$, $\Xg$, $o(\cdot)$, $s$, $[(\Xz)^\top \Xz]^{-1}$, $\frac{\D \betag}{\D \eta}$ and $\frac{\D \beta}{\D \eta}$.
		\State Compute timings of events $\Delta_{g}^\text{fuse}$, $\Delta_{g, k}^\text{split}$, $\Delta_{k}^\text{switch}$ and the minimum $\Delta^{\min}$ of them all.
                \While{$\eta^{(t)} + \Delta^{\min} < \eta_{\max}$}
		\State $\eta^{(t+1)} \leftarrow \eta^{(t)} + \Delta^{\min}$, $\beta^{(t+1)} \leftarrow \beta^{(t)} + \Delta^{\min} \frac{\D \beta}{\D \eta}$.
		\If{$\Delta_{g}^\text{fuse} = \Delta^{\min}$ for some $g$}
		\State Fuse $G_g$ and $G_{g+1}$, and update $\group$, $\betag$, $\Xg$, $o(\cdot)$, $[(\Xz)^\top \Xz]^{-1}$, $\frac{\D \betag}{\D \eta}$ and $\frac{\D \beta}{\D \eta}$. 
		\ElsIf{$\Delta_{0}^\text{switch} = \Delta^{\min}$}
		\State Switch the sign of $s_{o(1)}$. 
		\ElsIf{$\Delta_{k}^\text{switch} = \Delta^{\min}$ for some $k$}
                \State Switch the indices assigned to $o(k)$ and $o(k+1)$.  
		\ElsIf{$\Delta_{g, k}^\text{split} = \Delta^{\min}$ for some $g$ and $k$}
		\State Split $G_g$, and update $\group$, $\betag$, $\Xg$, $o(\cdot)$, $[(\Xz)^\top \Xz]^{-1}$, $\frac{\D \betag}{\D \eta}$ and $\frac{\D \beta}{\D \eta}$. 
		\EndIf
		\State $t\leftarrow t+1$
		\State Compute timings of events $\Delta_{g}^\text{fuse}$, $\Delta_{g, k}^\text{split}$, $\Delta_{k}^\text{switch}$ and the minimum $\Delta^{\min}$ of them all. 
		\EndWhile
                \If{$\eta_{\max}<\infty$}
		\State $\eta^{(t+1)} \leftarrow \eta_{\max}$, $\beta^{(t+1)} \leftarrow \beta^{(t)} + (\eta_{\max}-\eta^{(t)}) \frac{\D \beta}{\D \eta}$.
		\State $t\leftarrow t+1$
                \EndIf
	\end{algorithmic}
\end{algorithm}

\section{Quasi-spherical OSCAR}
A regularization sequence $\lambda_1,\lambda_2,\dots,\lambda_p$ for SLOPE can be arbitrarily defined as long as it is in a non-decreasing order, 
such that $0\le \lambda_1 \le \lambda_2 \le \cdots \le \lambda_p$.
Therefore, a variety of regularization sequences has been proposed for SLOPE.
This section presents a new type of regularization sequences for feature clustering. 

First, we review some designs of the regularization sequences in the previous studies before introducing the proposed type.
\citet{bogdan2015} presented two regularization sequence designs.
One is designed as a threshold sequence of the multiple testing procedure of nonzero coefficients by \citep{benjamini1995}
and defined as follows:
\begin{equation}\label{BH}
\lambda^\text{BH}_i = \eta \Phi^{-1}\kakko{1-q\frac{p-i+1}{2p}},
\end{equation}
where $\eta$ tunes the strength of penalty, $\Phi^{-1}$ is the inverse distribution function of the standard normal distribution, and
$q$ controls the target false discovery rate of nonzero coefficients in the SLOPE estimator.
\citet{kos2020} has proven some asymptotic results on the false descovery rate with this penalty under a Gaussian design, where all the $X$ entries are i.i.d. Gaussian.
The other type presented by \citet{bogdan2015} involves the previous one $\lambda_{BH}$ and recursively defined by $\lambda^\text{G}_p = \lambda^\text{BH}_p$ and
\begin{equation}\label{G}
\lambda^\text{G}_i = \eta \min\kakko{\lambda^\text{G}_{i+1}, \lambda^\text{BH}_i \sqrt{1+\frac{1}{n-p+i-3}\sum_{j>i}(\lambda^\text{G}_j)^2}}.
\end{equation}
This is designed to heuristically consider the influence of the sample correlation in $X$~\citep{bogdan2015}.

SLOPE is introduced as a generalization of OSCAR ~\citep{bondell2008simultaneous}.
OSCAR is originally defined by pairwise $L_\infty$ terms as
\begin{equation} \label{OSCAR}
\minimize_{\beta}\ \ \frac{1}{2} \nkakko{y-X\beta}^2 + \lambda_1 \sum_{i=1}^{p} \abs{\beta_i} + \lambda_2 \sum_{j<k} \max \ckakko{\abs{\beta_j}, \abs{\beta_k}},
\end{equation}
and its penalty terms can be rewritten by the SLOPE penalty $\sum_{i=1}^p \lambda^\text{OSCAR}_i |\beta |_{(i)}$ with
\begin{equation}\label{OS}
\lambda^\text{OSCAR}_i = q + \eta (i - 1),
\end{equation}
where $q$ and $\eta$ coincide with $\lambda_1$ and $\lambda2$ in \eqref{OSCAR}, respectively.
The contour surface of the penalty terms draws an octagonal shape, as in Figure~\ref{fig1}.
The solution of \eqref{OSCAR} is always the point where the contour surface of the least-square term contacts with that of the penalty terms.
Therefore, the octagonal structure encourages not only sparsity, but equality of absolute coefficients with close least-square solutions and multicollinearity.

However, the octagonal structure in OSCAR is much weakened for high-dimensional problems.
For illustration, consider the plane of $\beta_1$ and $\beta_2$, where the other coefficients $\beta_3,\dots,\beta_p$ are all zero.
Figure~\ref{fig1} shows that the upper-right vertex of the dashed contour line $\sum_{i=1}^p \lambda^\text{OSCAR}_i |\beta |_{(i)} = 1$ 
with $\lambda^\text{OSCAR}_i = i - 1$ and $\beta_3 = \cdots = \beta_p = 0$ is at $(\frac{p}{2p-1},\frac{p}{2p-1})^\top$,
which rapidly converges to $(1/2,1/2)^\top$ and reduces the octagonal shape to a diamond as $p$ increases.

\begin{figure}[htbp]
  \begin{center}
	\includegraphics{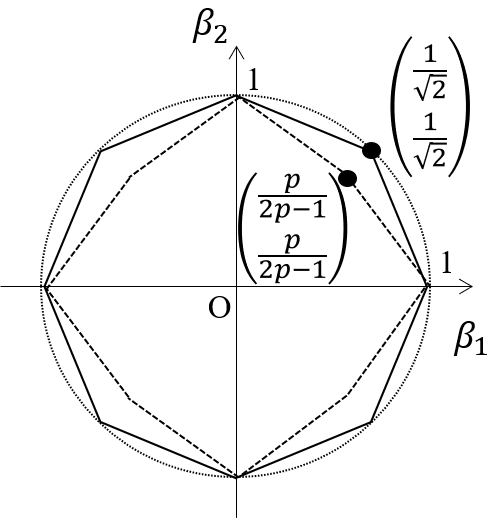}
  \end{center}
  \caption{Contour lines of $\sum_{i=1}^p \lambda_i |\beta |_{(i)} = 1$ for the OSCAR penalty $\lambda^\text{OSCAR}$ (dashed line) and the QS-OSCAR penalty $\lambda^\text{QS}$ (solid line)
given $\beta_3 = \cdots = \beta_p = 0$.}
  \label{fig1}
\end{figure}

For feature clustering, the octagonal structure should be preserved for high-dimensional problems.
This can be realized by setting the regularization sequence $\lambda$ such that all the vertices of the contour surface are located on the sphere of the same radius.
Such a sequence, which we call the quasi-spheric and octagonal shrinkage and clustering algorithm for regression (QS-OSCAR), is defined by
\begin{equation}\label{QS}
\lambda^\text{QS}_i = \eta (\sqrt{i}-\sqrt{i-1}).
\end{equation}
Figure~\ref{fig1} illustrates that the contour of this penalty where $\beta_3 = \cdots = \beta_p = 0$ always draws a regular octagon
whose vertices are on the same circle for any $p$.
Although the contour surface of the penalty $\sum_{i=1}^p \lambda^\text{QS}_i |\beta |_{(i)}$ is not a sphere,
the surface with $\lambda^\text{QS}$ is the most similar to a sphere among all possible $\lambda$ in view of the
ratio of the inradius and circumradius of the contour surface according to the following theorem:

\begin{theorem} \label{theorem2}
The ratio of $\max_{\| \beta \| = r} \sum_{i=1}^p \lambda_i |\beta|_{(i)}$ from $\min_{\| \beta \| = r} \sum_{i=1}^p \lambda_i |\beta|_{(i)}$
takes its mininum $\rho_p = \sqrt{\sum_{i=1}^p (\sqrt{i}-\sqrt{i-1})^2} = O(\sqrt{\log p})$ 
if and only if $\lambda = \lambda^\text{QS}$ in \eqref{QS} for some $\eta >0$.
\end{theorem}

\begin{proof}
From the Cauchy-Schwarzs inequality, we have $\sum_{i=1}^p \lambda_i |\beta|_{(i)} \le  \sqrt{\sum_{i=1}^p \lambda_i^2}\sqrt{\sum_{i=1}^p |\beta|_{(i)}^2} = \| \lambda \| \| \beta \|$,
where the equality holds if and only if $|\beta|_{(i)} = r \lambda_i/\| \lambda \| $ for $r>0$ and $i=1,\dots,p$.
Therefore, we have $\max_{\| \beta \| = r} \sum_{i=1}^p \lambda_i |\beta|_{(i)}= r\| \lambda \|$.

Next, we evaluate $\min_{\| \beta \| = r} \sum_{i=1}^p \lambda_i |\beta|_{(i)}$.
We obtain $\min_{\| \beta \| = r} \sum_{i=1}^p \lambda_i |\beta|_{(i)} = \min_{\beta \in \mathcal{B}} \langle \lambda,\beta \rangle$,
where $\mathcal{B} = \{\beta; \| \beta \| = r, 0\le \beta_1 \le \cdots \le \beta_p\}$.
The set $\mathcal{B}$ is also represented by $\mathcal{B} = \{ \beta = \frac{r \sum_{i=1}^p c_i b_i}{\| \sum_{i=1}^p c_i b_i \|}; 0\le c_1,\dots,c_p \le 1, \sum_{i=1}^p c_i = 1\}$,
where $b_1 = (0,\dots,0,1)^\top,\, b_2 = (0,\dots,0,\frac{1}{\sqrt{2}},\frac{1}{\sqrt{2}})^\top,$\\$\dots, b_p = (\frac{1}{\sqrt{p}},\dots,\frac{1}{\sqrt{p}})^\top$.
Subsequently, since $rb_1,\dots,rb_p \in \mathcal{B}$, we obtain 
\begin{equation*}
\min_{\beta \in \mathcal{B}} \langle \lambda,\beta \rangle \le \min_{i=1,\dots,p}  \langle \lambda,rb_i \rangle 
= \min_{\beta \in \mathcal{B}^\prime} \langle \lambda,\beta \rangle, 
\end{equation*}
where $\mathcal{B}^\prime = \{ \beta = r \sum_{i=1}^p c_i b_i; 0\le c_1,\dots,c_p \le 1, \sum_{i=1}^p c_i = 1\}$.
In contrast, since $\| \sum_{i=1}^p c_i b_i \| \le 1$, we have
$\langle \lambda,\frac{r \sum_{i=1}^p c_i b_i}{\| \sum_{i=1}^p c_i b_i \|} \rangle \ge \langle \lambda,r\sum_{i=1}^p c_i b_i \rangle$
 for any $0\le c_1,\dots,c_p \le 1, \sum_{i=1}^p c_i = 1$; hence,
\begin{equation*}
\min_{\beta \in \mathcal{B}} \langle \lambda,\beta \rangle 
\ge \min_{\beta \in \mathcal{B}^\prime} \langle \lambda,\beta \rangle.
\end{equation*}
Consequently, it holds that $\min_{\| \beta \| = r} \sum_{i=1}^p \lambda_i |\beta|_{(i)} = \min_{\beta \in \mathcal{B}} \langle \lambda,\beta \rangle = \min_{\beta \in \mathcal{B}^\prime} \langle \lambda,\beta \rangle$.

Let $\check{\lambda} = \lambda^\text{QS}/\| \lambda^\text{QS} \| = (\sqrt{p}-\sqrt{p-1},\sqrt{p-1}-\sqrt{p-2},\dots,\sqrt{2}-1,1)^\top /\rho_p$
be the normalized weights of $\lambda^\text{QS}$.
When $\lambda \propto \check{\lambda}$, we obtain $\langle \lambda,rb_i \rangle = r \| \lambda \|/\rho_p$ for $i=1,\dots,p$,
and hence $\min_{\| \beta \| = r} \sum_{i=1}^p \lambda_i |\beta|_{(i)} = \min_{\beta \in \mathcal{B}^\prime} \langle \lambda,\beta \rangle = r \| \lambda \|/\rho_p$.
Furthermore, since $\check{\beta} = r\check{\lambda}/\rho_p \in \mathcal{B}^\prime$, we obtain
$\langle \lambda,\check{\beta} \rangle \le r \| \lambda \|/\rho_p$ from the Cauchy-Schwarzs inequality,
where the equality holds if and only if $\lambda \propto \check{\lambda}$.

Above all, we have $\min_{\| \beta \| = r} \sum_{i=1}^p \lambda_i |\beta|_{(i)} \le r \| \lambda \|/\rho_p$, and hence
$\frac{\max_{\| \beta \| = r} \sum_{i=1}^p \lambda_i |\beta|_{(i)}}{\min_{\| \beta \| = r} \sum_{i=1}^p \lambda_i |\beta|_{(i)}} \ge \frac{r \| \lambda \|}{r \| \lambda \|/\rho_p} = \rho_p$,
where the equality holds if and only if $\lambda = \lambda^\text{QS} \propto \check{\lambda}$ for some $\eta$.

Finally, as $\sqrt{i}-\sqrt{i-1} = \sqrt{i}(1-\sqrt{1-\frac{1}{i}}) = \frac{1}{2\sqrt{i}} +  o(\frac{1}{\sqrt{i}})$,
we obtain
\begin{align*}
\rho_p^2 &= \sum_{i=1}^p (\sqrt{i}-\sqrt{i-1})^2 \\
&= 1 + \sum_{i=2}^p \dkakko{\frac{1}{2\sqrt{i}} +  o\kakko{\frac{1}{\sqrt{i}}}} \times (\sqrt{i}-\sqrt{i-1})\\
&\le 1 + \int_1^p \dkakko{\frac{1}{2\sqrt{z}} +  o\kakko{\frac{1}{\sqrt{z}}}}d\sqrt{z} \\
&= 1 + \log p + o(\log p)\\
&= O(\log p).
\end{align*}
Consequently, $\rho_p = O(\sqrt{\log p})$, which completes the proof.
\end{proof}

\begin{remark}
From Theorem~\ref{theorem2}, 
the ratio of the circumradius and inradius of the surface $\sum_{i=1}^p \lambda_i |\beta|_{(i)} = 1$
takes its mininum $\rho_p$ if and only if $\lambda = \lambda^\text{QS}$ for some $\eta >0$.
It implies that the surface becomes the most similar to a sphere when $\lambda = \lambda^\text{QS}$.
Although the surface becomes less similar to a sphere as $p$ increases, it is not fast because $\rho_p = O(\sqrt{\log p})$.
Figure~\ref{fig2} shows that $\rho_p$ is only approximately $1.47$ for $p = 100$ and approximately $1.82$ even for $p=10000$.

In contrast, the circumsphere and insphere of the contour surface $\sum_{i=1}^p \lambda_i |\beta|_{(i)} = 1$ with $\lambda = \lambda^\text{OSCAR}$
contact with the surface at $\beta = (0,\dots,0,1/\lambda^\text{OSCAR}_p)^\top$ and $\beta = \lambda^\text{OSCAR}/\| \lambda^\text{OSCAR} \|^2$, respectively.
Therefore, the ratio of the circumradius and inradius of the contour surface in OSCAR is
$\| \lambda^\text{OSCAR} \|/\lambda^\text{OSCAR}_p \gtrsim \sqrt{p}$, which goes to infinity much faster than QS-OSCAR.
\end{remark}

\begin{figure}[htbp]
  \begin{center}
	\includegraphics{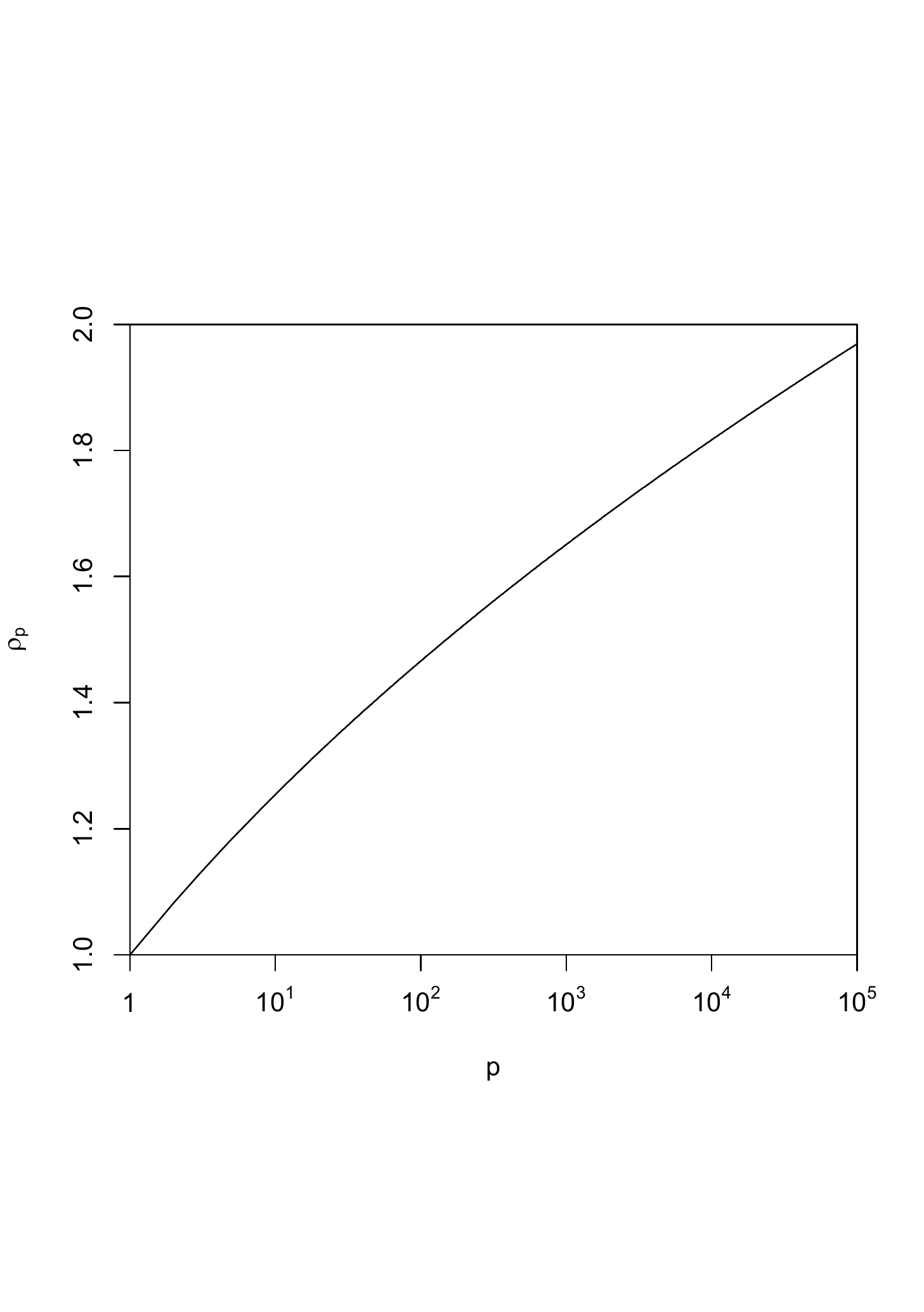}
  \end{center}
  \caption{Development of $\rho_p$ with respect to $p$.}
  \label{fig2}
\end{figure}

\section{Simulation Studies}
In this section, we compare some features of the various regularization sequences for SLOPE introduced in the previous section
through numerical experiments with synthetic datasets.
All experiments were conducted on a Windows 10 64-bit machine with an Intel i7-8665U CPU at 1.90 GHz and 16 GB of RAM.

The synthetic datasets were generated from the model $y=X\beta + e$, where $e \sim N(0, I_n)$.
The problem size $(p,n)$ was set from four levels $\{ (20,200),(40,400),(80,800),(160,1600) \}$.
There are two scenarios for the true coefficients $\beta$ and covariates $X$.
The first scenario involves $\beta$ generated by $\beta = (\theta,-\theta)^\top$, where $\theta \sim N(0, I_{0.5p})$,
and each row $x_i$ of $X$ generated independently by $x_i \sim N(0, \Sigma)$, where
\begin{equation*}
\Sigma = \frac{1}{\sqrt{n}}\left( \begin{array}{cc} I_{0.5p} & 0.8I_{0.5p} \\  0.8I_{0.5p} & I_{0.5p} \end{array} \right).
\end{equation*}
This covariance matrix $\Sigma$ means that the covariates of $\beta_i = \theta_i$ and $\beta_{i+0.5p} = -\theta_i$ are correlated as 
$Corr(x_{i,j},x_{i+0.5p,j}) = 0.8$ for $i=1,\dots,0.5p$ and $j=1,\dots,n$.
The second scenario involves each component of $\beta$ independently sampled from $\{-2,-1,0,1,2\}$, and
each element of $X$ independently sampled from $\{-1,0,1\}$.
For each setting, 100 datasets are generated to evaluate the average performance.
From each dataset, 
four solution paths of SLOPE along $\eta \in [0,\infty)$ with $\lambda = \lambda^\text{BH}$, $\lambda^\text{G}$, $\lambda^\text{OSCAR}$ and $\lambda^\text{QS}$, respectively,
were calculated through Algorithm~\ref{algorithm}.

Tables~\ref{table1} and \ref{table2} show the mean number of nonzero coefficients in $\beta(\eta^{(t)})$ for $t=1,\dots,T$ in the respective scenarios,
where $T$ is the number of iterations in Algorithm~\ref{algorithm}.
In both scenarios, the number of nonzero coefficients tends to be smaller in $ \lambda^\text{BH}$ and $\lambda^\text{G}$ than the others.
The mean number of zeros in $\beta(\eta^{(t)})$ is only a few in $\lambda^\text{OSCAR}$.

Tables~\ref{table3} and \ref{table4} list the mean numbers of nonzero groups in $\beta(\eta^{(t)})$ for $t=1,\dots,T$.
The mean numbers of nonzero groups with $\lambda^\text{QS}$ are about half of the dimension $p$ and the smallest of the four regularization sequences in both scenarios.
In contrast, the mean numbers of nonzero groups are not much different among $ \lambda^\text{BH}$, $\lambda^\text{G}$ and $ \lambda^\text{OSCAR}$.

Tables~\ref{table5} and \ref{table6} show the mean number of fusing/splitting events that occurred in Algorithm~\ref{algorithm}.
The switching events in the algorithm are not counted because they do not change the slope of the solution paths, and their runtimes are much shorter than the other events.
The number of events increases by approximately four times as the dimension $p$ is doubled in each case, 
which implies that the number of events is approximately $O(p^2)$ for these scenarios.
Note that evaluating the number of iterations in general cases is difficult and might be exponentially large (see for example, \citep{mairal2012}).
The numbers of events are much larger in Scenario 1 than in Scenario 2, which may reflect the complexity of the coefficients $\beta$ and the covariates $X$.
The large numbers of events in high-dimensional regression make it difficult to detect the change points in a solution path exhaustively using grid search methods.
Hence, the proposed solution path algorithm that tracks all the breaking points is needed for fine-tuning of $\eta$ and feature grouping.

\begin{table}[htbp]
  \caption{Mean number of nonzero coefficients along a solution path in Scenario 1.}
  \label{table1}
  \centering
  \begin{tabular}{crrrr}
    \toprule
    Regularization sequence $\lambda$ & \multicolumn{4}{c}{Problem size $(p,n)$} \\
    \cmidrule(r){2-5}
     & $(20,200)$ & $(40,400)$ & $(80,800)$ & $(160,1600)$ \\
    \midrule
    Benjamini and Hochberg $\lambda^\text{BH}$& 14.3  & 29.6 &  59.8 & 120.9\\
    Gaussian $\lambda^\text{G}$              & 14.0  & 28.9 &  58.2 & 118.0\\
    OSCAR $\lambda^\text{OSCAR}$                 & 18.6  & 38.0 &  77.1 & 155.6\\
    Quasi-spheric OSCAR $\lambda^\text{QS}$   & 16.6  & 33.7 &  67.9 & 137.0\\
    \bottomrule
  \end{tabular}
\end{table}

\begin{table}[htbp]
  \caption{Mean number of nonzero coefficients along a solution path in Scenario 2.}
  \label{table2}
  \centering
  \begin{tabular}{crrrr}
    \toprule
    Regularization sequence $\lambda$& \multicolumn{4}{c}{Problem size $(p,n)$} \\
    \cmidrule(r){2-5}
     & $(20,200)$ & $(40,400)$ & $(80,800)$ & $(160,1600)$ \\
    \midrule
    Benjamini and Hochberg $\lambda^\text{BH}$& 15.2 &  31.5 &  64.6 & 130.1\\
    Gaussian $\lambda^\text{G}$              & 14.8 &  30.7 &  63.1 & 127.2\\
    OSCAR $\lambda^\text{OSCAR}$                 & 18.8 &  38.3 &  78.3 & 157.5\\
    Quasi-spheric OSCAR $\lambda^\text{QS}$   & 17.5 &  35.3 &  71.5 & 142.9\\
    \bottomrule
  \end{tabular}
\end{table}

\begin{table}[htbp]
  \caption{Mean number of nonzero fused groups along a solution path in Scenario 1.}
  \label{table3}
  \centering
  \begin{tabular}{crrrr}
    \toprule
    Regularization sequence $\lambda$& \multicolumn{4}{c}{Problem size $(p,n)$} \\
    \cmidrule(r){2-5}
     & $(20,200)$ & $(40,400)$ & $(80,800)$ & $(160,1600)$ \\
    \midrule
    Benjamini and Hochberg $\lambda^\text{BH}$& 12.7 &  26.7 &  54.1 &109.7\\
    Gaussian $\lambda^\text{G}$              & 13.0 &  27.1 &  54.7 &111.1\\
    OSCAR $\lambda^\text{OSCAR}$                 & 12.4 &  26.1 &  53.4 &108.1\\
    Quasi-spheric OSCAR $\lambda^\text{QS}$   & 10.9 &  23.0 &  47.3 & 96.7\\
    \bottomrule
  \end{tabular}
\end{table}

\begin{table}[htbp]
  \caption{Mean number of nonzero fused groups along a solution path in Scenario 2.}
  \label{table4}
  \centering
  \begin{tabular}{crrrr}
    \toprule
    Regularization sequence $\lambda$& \multicolumn{4}{c}{Problem size $(p,n)$} \\
    \cmidrule(r){2-5}
     & $(20,200)$ & $(40,400)$ & $(80,800)$ & $(160,1600)$ \\
    \midrule
    Benjamini and Hochberg $\lambda^\text{BH}$& 12.8 &  26.2 &  53.4 & 106.8\\
    Gaussian $\lambda^\text{G}$              & 13.2 &  27.2 &  55.6 & 111.4\\
    OSCAR $\lambda^\text{OSCAR}$                 & 12.3 &  25.8 &  53.0 & 106.1\\
    Quasi-spheric OSCAR $\lambda^\text{QS}$   & 11.2 &  22.5 &  45.3 &  90.3\\
    \bottomrule
  \end{tabular}
\end{table}

\begin{table}[htbp]
  \caption{Mean number of fusing/splitting events along a solution path in Scenario 1.}
  \label{table5}
  \centering
  \begin{tabular}{crrrr}
    \toprule
    Regularization sequence $\lambda$& \multicolumn{4}{c}{Problem size $(p,n)$} \\
    \cmidrule(r){2-5}
     & $(20,200)$ & $(40,400)$ & $(80,800)$ & $(160,1600)$ \\
    \midrule
    Benjamini and Hochberg $\lambda^\text{BH}$& 140   &      526    &    2029  &      8079\\
    Gaussian $\lambda^\text{G}$              & 139   &      524    &    2031  &      8118\\
    OSCAR $\lambda^\text{OSCAR}$                 & 138   &      496    &    1887  &      7462\\
    Quasi-spheric OSCAR $\lambda^\text{QS}$   & 178   &      656    &    2414  &      9455\\
    \bottomrule
  \end{tabular}
\end{table}

\begin{table}[htbp]
  \caption{Mean number of fusing/splitting events along a solution path in Scenario 2.}
  \label{table6}
  \centering
  \begin{tabular}{crrrr}
    \toprule
    Regularization sequence $\lambda$ & \multicolumn{4}{c}{Problem size $(p,n)$} \\
    \cmidrule(r){2-5}
     & $(20,200)$ & $(40,400)$ & $(80,800)$ & $(160,1600)$ \\
    \midrule
    Benjamini and Hochberg $\lambda^\text{BH}$&68    &     245    &     968   &     3799 \\
    Gaussian $\lambda^\text{G}$              &75    &     271    &    1073   &     4264 \\
    OSCAR $\lambda^\text{OSCAR}$                 &51    &     178    &     710   &     2788 \\
    Quasi-spheric OSCAR $\lambda^\text{QS}$   &52    &     182    &     697   &     2743 \\
    \bottomrule
  \end{tabular}
\end{table}

\section{Conclusion}
This study proposed a solution path algorithm that yields an entire exact solution path for SLOPE, which is extended from the path algorithm for OSCAR~\citep{takahashi2020}.
The algorithm has three types of events: 
a fusing event for a pair of feature groups, a splitting event of a group, and a switching event that affects the next splitting time.
The next timings of the splitting events are derived from a simple optimality condition given in Theorem~\ref{theorem}.

A new design for regularization weights $\lambda$, called QS-OSCAR, is proposed for feature clustering in Section 3.
The simulation study demonstrated that QS-OSCAR groups absolute coefficients more efficiently than OSCAR and other penalties for SLOPE.
Since the octagonal contour surface of QS-OSCAR is close to a sphere for a reasonably high dimension,
QS-OSCAR may behave like $L_2$ regularization for shrinking and encouraging the similarity of the absolute coefficients of the highly correlated predictors.
In contrast, unlike $L_2$ regularization, QS-OSCAR also plays a role in dimensionality reduction by promoting the sparsity and grouping of coefficients.
Additionally, we can obtain an entire solution path for QS-OSCAR using the proposed algorithm.


\begin{thebibliography}{}

	\bibitem[Bao et~al., 2019]{bao2019}
	Bao, R., Gu, B., and Huang, H. (2019),
	\newblock ``Efficient approximate solution path algorithm for order weight l\_1-norm with accuracy guarantee,''
	\newblock In {\em 2019 IEEE international conference on data mining}, 958--963.
	
	\bibitem[Bao et~al., 2020]{bao2020fast}
	Bao, R., Gu, B., and Huang, H. (2020),
	\newblock ``Fast oscar and owl regression via safe screening rules,''
	\newblock In {\em Proceedings of the 37th international conference on machine learning}, 1168--1178.
	
	\bibitem[Bellec et~al., 2018]{bellec2018}
	Bellec, P.~C., Lecu\'e, G., and Tsybakov, A.~B. (2018),
	\newblock ``Slope meets lasso: improved oracle bounds and optimality,''
	\newblock {\em Annals of Statistics}, 46(6B), 3603--3642.
	
	\bibitem[Benjamini and Hochberg, 1995]{benjamini1995}
	Benjamini, Y., and Hochberg, Y. (1995),
	\newblock ``Controlling the false discovery rate: a practical and powerful approach to multiple testing,''
	\newblock {\em Journal of the Royal Statistical Society. Series B}, 57(1), 289--300.
	
	\bibitem[Bogdan et~al., 2013]{bogdan2013}
	Bogdan, M., van~den Berg, E., Su, W., and Cand\'es, E.~J. (2015),
	\newblock ``Statistical estimation and testing via the sorted l1 norm,''
	\newblock {\em arXiv preprint arXiv:1310.1969}.
	
	\bibitem[Bogdan et~al., 2015]{bogdan2015}
	Bogdan, M., van~den Berg, E., Sabatti, C., Su, W., and Cand\'es, E.~J. (2015),
	\newblock ``Slope-adaptive variable selection via convex optimization,''
	\newblock {\em Annals of Applied Statistics}, 9(3), 1103--1140.
	
	\bibitem[Bondell and Reich, 2008]{bondell2008simultaneous}
	Bondell, H.~D., and Reich, B.~J. (2008),
	\newblock ``Simultaneous regression shrinkage, variable selection, and supervised
	  clustering of predictors with oscar,''
	\newblock {\em Biometrics}, 64(1), 115--123.
	
	\bibitem[Hardy et~al., 1934]{hardy1934}
	Hardy, G.~H., Littlewood, J.~E., and P\'olya, G. (1934),
	\newblock {\em Inequalities},
	\newblock Cambridge University Press.
	
	\bibitem[Kos and Bogdan, 2020]{kos2020}
	Kos, M., and Bogdan, M. (2020),
	\newblock ``On the asymptotic properties of slope,''
	\newblock {\em Sankhy\={a} A: The Indian Journal of Statistics}, 82(2), 499-532.

	\bibitem[Larsson et~al., 2020]{larsson2020}
	Larsson, J., Bogdan, M., and Wallin, J. (2020),
	\newblock ``The strong screening rule for slope,''
	\newblock {\em arXiv preprint arXiv:2005.03730}.

	\bibitem[Lecu\'e and Mendelson, 2018]{lecue2018}
	Lecu\'e, G., and Mendelson, S. (2018),
	\newblock ``Regularization and the small-ball method i: sparse recovery,''
	\newblock {\em Annals of Statistics}, 46(2), 611--641.
	
	\bibitem[Luo et~al., 2019]{luo2019efficient}
	Luo, Z., Sun, D., Toh, K. C., and Xiu, N. (2019),
	\newblock ``Solving the oscar and slope models using a semismooth newton-based
	  augmented lagrangian method,''
	\newblock {\em Journal of Machine Learning Research}, 20(106), 1--25.
	
	\bibitem[Mairal and Yu, 2012]{mairal2012}
	Mairal, J., and Yu, B. (2012),
	\newblock ``Complexity analysis of the lasso regularization path,''
	\newblock In {\em Proceedings of the 29th international conference on machine learning}, 1168--1178.
	
	\bibitem[Rosset and Zhu, 2007]{rosset2007}
	Rosset, S., and Zhu, J. (2007),
	\newblock ``Piecewise linear regularized solution paths,''
	\newblock {\em Annals of Statistics}, 35(3), 1012--1030.
	
	\bibitem[Su and Cand\'es, 2016]{su2016}
	Su, W., and Cand\'es, E.~J. (2016),
	\newblock ``Slope is adaptive to unknown sparsity and asymptotically minimax,''
	\newblock {\em Annals of Statistics}, 44(3), 1038--1068.
	
	\bibitem[Takahashi and Nomura, 2020]{takahashi2020}
	Takahashi, A., and Nomura, S. (2020),
	\newblock ``Efficient path algorithms for clustered lasso and oscar,''
	\newblock {\em arXiv preprint arXiv:2006.08965}.

	\bibitem[Tibshirani, 1996]{tibshirani1996regression}
	Tibshirani, R. (1996),
	\newblock ``Regression shrinkage and selection via the lasso,''
	\newblock {\em Journal of the Royal Statistical Society. Series B}, 58(1), 267--288.
	
	\bibitem[Tibshirani et~al., 2005]{tibshirani2005sparsity}
	Tibshirani, R., Saunders, M., Rosset, S., Zhu, J., and Knight, K. (2005),
	\newblock ``Sparsity and smoothness via the fused lasso,''
	\newblock {\em Journal of the Royal Statistical Society: Series B}, 67(1), 91--108.
	
	\bibitem[Tibshirani and Taylor, 2011]{tibshirani2011}
	Tibshirani, R.~J., and Taylor, J. (2011),
	\newblock ``The solution path of the generalized lasso,''
	\newblock {\em Annals of Statistics}, 39(3), 1335--1371.
	
	\bibitem[Zeng and Figueiredo, 2014]{zeng2014}
	Zeng, X., and Figueiredo, M.~A. (2014),
	\newblock ``Decreasing weighted sorted $l_1$ regularization,''
	\newblock {\em IEEE Signal Processing Letters}, 21(10), 1240--1244.

\end{thebibliography}
\end{document}